\documentclass[11pt]{article}
\usepackage{amssymb,amsthm,amsmath,float,epsfig}
\usepackage[small,width=1.2\linewidth]{caption}
\usepackage[total={6.5in,9.0in}, top=1.0in, left=1.0in, includefoot]{geometry}
\usepackage{epsfig}
\usepackage{graphicx}
\usepackage{enumerate}
\usepackage{enumitem}
\usepackage{bm}
\usepackage{authblk}
\usepackage{setspace}
\usepackage{color}

\numberwithin{equation}{section}
\newtheorem*{definition}{Definition}
\newtheorem*{theorem}{Theorem}
\newtheorem*{lemma}{Lemma}
\newtheorem*{remark}{Remark}

\newcommand{\appropto}{\mathrel{\vcenter{
  \offinterlineskip\halign{\hfil$##$\cr
    \propto\cr\noalign{\kern2pt}\sim\cr\noalign{\kern-2pt}}}}}

\graphicspath{{figures/}}

\title{Dynamics of Nonlinear Random Walks on Complex Networks}

\author{Per Sebastian Skardal\thanks{persebastian.skardal@trincoll.edu} }
\author{Sabina Adhikari}
\affil{Department of Mathematics, Trinity College, Hartford, CT 06106, USA}
\date{}

\begin{document}

\maketitle

\begin{abstract}
In this paper we study the dynamics of nonlinear random walks. While typical random walks on networks consist of standard Markov chains whose static transition probabilities dictate the flow of random walkers through the network, nonlinear random walks consist of nonlinear Markov chains whose transition probabilities change in time depending on the current state of the system. This framework allows us to model more complex flows through networks that may depend on the current system state. For instance, under humanitarian or capitalistic direction, resource flow between institutions may be diverted preferentially to poorer or wealthier institutions, respectively. Importantly, the nonlinearity in this framework gives rise to richer dynamical behavior than occurs in linear random walks. Here we study these dynamics that arise in weakly and strongly nonlinear regimes in a family of nonlinear random walks where random walkers are biased either towards (positive bias) or away from (negative bias) nodes that currently have more random walkers. In the weakly nonlinear regime we prove the existence and uniqueness of a stable stationary state fixed point provided that the network structure is primitive that is analogous to the stationary distribution of a typical (linear) random walk. We also present an asymptotic analysis that allows us to approximate the stationary state fixed point in the weakly nonlinear regime. We then turn our attention to the strongly nonlinear regime. For negative bias we characterize a period-doubling bifurcation where the stationary state fixed point loses stability and gives rise to a periodic orbit below a critical value. For positive bias we investigate the emergence of multistability of several stable stationary state fixed points.
\end{abstract}

\section{Introduction}\label{sec:01}
Discrete-time random walks on complex networks represent a powerful tool, surfacing in a wide range of applications~\cite{Noh2004PRL}. Perhaps the most famous among these applications is Google's PageRank algorithm, where the PageRank centrality of a given webpage is given by the fraction of time a random walker spends at that webpage in the long run, i.e., the steady-state dynamics of the probability distribution of occupation~\cite{Brin1998,Page1999}. PageRank itself has found applications in many venues beyond webpage centralities and centrality rankings in other networks; a handful of these include identifying correlated genes and proteins, studying traffic flow, and predicting outcomes of sports matches~\cite{Gleich2015SIAMRev}. Other widely-used applications of random walks on networks include community detection methods~\cite{Rosvall2008PNAS} and modeling transport through large systems~\cite{Nicosia2017PRL,Gorenflo2002CP}.

An important reason why random walks constitute such a powerful tool in such a wide range of applications is due to the simplicity of the underlying dynamics. Consider a network, or graph, $G=(V,E)$, where $V$ is the set of nodes, or vertices, and $E$ is the set of links, or edges, describing pair-wise connections between distinct vertices. Such a network may be described by an adjacency matrix $A$ with entries $[A]_{ij}=a_{ij}$, where $a_{ij}=1$ if nodes $i$ and $j$ are connected with a link, and $a_{ij}=0$ otherwise. This framework assumes that the underlying network is undirected and unweighted (i.e., binary), but can be generalized to the directed and weighted case if, respectively, $a_{ij}\ne a_{ji}$ for some pair $(i,j)$ or non-zero entries of $A$ may take on non-unit values to indicate a ``strength''. (Here we focus our attention on the undirected, unweighted case.) Next, consider a single random walker on the network and the probability $p_i(t)\in[0,1]$ that the walker occupies node $i$ at time $t$. The dynamics of the random walk, characterized by the evolution of the probability vector $\bm{p}(t)$, are given by
\begin{align}
p_i(t+1)=\sum_{j=1}^N\pi_{ij}p_j(t),\hskip2ex\text{or in vector form,}\hskip2ex\bm{p}(t+1)=\Pi\bm{p}(t),\label{eq:01:01}
\end{align}
where $\Pi$ is the transition matrix whose entries $[\Pi]_{ij}=\pi_{ij}$ represent the probability of a random walker moving to node $i$ at time $t+1$ given that it occupies node $j$ at time $t$.  In the most typical case, the transition probabilities are {\it unbiased}, meaning the probabilities of a random walker moving from node $j$ to any of the neighbors of node $j$ are uniform. Since $\pi_{ij}$ is the conditional probability of moving from $j\to i$ in one time step, $\Pi$ must be column stochastic, i.e., $\sum_{i=1}^N\pi_{ij}=1$, this corresponds to the choice $\pi_{ij}=a_{ij}/k_{j}$, where $k_{j}=\sum_{i=1}^Na_{ij}$ is the degree of node $j$. Generalizations to this choice of transition probabilities have also been studied, biasing random walkers towards or away from nodes with certain characteristics, usually related to the local network structure~\cite{Gomez2008PRE}. In particular, it has been shown that appropriately biasing random walks increases the efficiency with which random walks traverse and explore their networks~\cite{Sinatra2011PRE}.

The simple structure of the discrete-time random walk on a network, which consists of an $N$-state discrete Markov chain, yields relatively simple dynamics due to the linearity of Eq.~(\ref{eq:01:01}). Assuming that the transition matrix $\Pi$ is primitive, i.e., irreducible and aperiodic, the dynamics of Eq.~(\ref{eq:01:01}) converge to a unique stationary state fixed point, denoted $\bm{p}^*$, that satisfies $\bm{p}^*=\Pi \bm{p}^*$~\cite{Durrett2016}. We note that this condition is equivalent to those that asserts the existence of a unique, real, largest eigenvalue of $\Pi$ with an associated positive eigenvector using the Perron-Frobenius theorem~\cite{MacCluer2000SIAMRev}. From this point of view, the dynamics can be understood as the convergence of $\bm{p}(t)=\Pi^t\bm{p}(0)$ to the stationary distribution $\bm{p}^*$, which is precisely the dominant eigenvector of the transition matrix $\Pi$.

In this work we study the dynamics that arise in {\it nonlinear random walks} on networks, which consist of an $N$-state discrete nonlinear Markov chain~\cite{Frank2013ISRN,Kolokoltsov2010}. In this more complicated setting the transition probabilities that dictate the flow of random walkers through he network change in time as a function of the current system state. The dynamics of nonlinear random walks are important for their ability to model flows through networks that are more complicated than the linear case. Consider, for instance, the flow of resources through a network of institutions. In such a scenario an individual random walker may represents one unit of the resource, with a large number $N\gg1$ of random walkers moving simultaneously through the network such that $p_i(t)$ represents not only the probability of a given random walker being present at node $i$ at time $t$, but also the fraction of all random walkers that occupy node $i$ at time $t$. In such a scenario, we argue that modeling the movement of resources using a random walk with static transition probabilities may be too simplistic and unrealistic. Consider, for instance, the flow of resources with a humanitarian purpose. In this scenario, resources are diverted more prevalently to locations that are currently lacking in resources compared to locations that are more wealthy. In terms of a random walk, this can be modeled by modifying transition probabilities to bias random walkers to nodes that have fewer random walkers present. On the other hand, financial agents with a more capitalist philosophy tend to invest more of their resources in institutions that are already wealthy, making for a safe investment. In terms of a random walk, this can be modeled by modifying transition probabilities to bias random walkers towards nodes that have more random walkers present. Notice that in either case, to more realistically model the dynamics, we require that the transition probabilities are not static, but depend on the current state of the system itself. 

In order to facilitate the kind of behavior described above, transition probabilities in nonlinear random walks are allowed to be non-static, i.e., changing in time in response to the current state of the system, so we let $\Pi=\Pi(\bm{p})$. Note that, as soon as the transition matrix $\Pi$ depends on the probability vector $\bm{p}$, the linearity of the dynamics is broken, raising several important questions. For instance, what becomes of the stationary state observed in the linear random walk? Depending on the nature of the nonlinearity, does a similar unique, stationary state exist? If so, how does the nonlinearity modify it? What other, more exotic, dynamics are observed? How do the dynamics make transitions between simple and complicated behaviors? In fact, the nonlinear random walk described above constitutes particular instance of a nonlinear Markov chain~\cite{Frank2013ISRN,Kolokoltsov2010}. Some properties of nonlinear Markov chains have to date been investigated, for instance ergodicity~\cite{Butkovsky2014SIAM,Saburov2016NLA}, the emergence of multiple solutions~\cite{Frank2008PLA}, their role in game theory~\cite{Kolokoltsov2012IJSP}, and other modeling applications~\cite{Frank2008JPA,Frank2009EPJB,Frank2011BJP}. However, little work has studied the behavior that emerges in a nonlinear random walk from a nonlinear dynamics perspective as a result of the rules that govern the transition probabilities and the underlying network structure. In this paper we take a nonlinear dynamics approach to studying the behavior of nonlinear random walks on networks.

This paper is organized as follows. In Sec.~\ref{sec:02} we define nonlinear random walks on complex networks, present some mathematical preliminaries, and illustrate the dynamics of the nonlinear random walk on a network. In Sec.~\ref{sec:03} we begin our investigation of the dynamics in the weakly nonlinear regime. In particular, we present and prove a theorem asserting the existence and uniqueness of a stable stationary state fixed point for sufficiently weak nonlinearity and appropriate network structures. In particular, we require that (like the case of the linear random walk) the network structure is primitive. In Sec.~\ref{sec:04} we study the structure of this stationary state, using a perturbative analysis to derive an asymptotic approximation for the stationary state fixed point in the weakly nonlinear regime. In Sec.~\ref{sec:05} we turn our attention beyond the weakly nonlinear regime and consider dynamics in the strongly nonlinear regime. We begin by characterizing a period-doubling bifurcation that occurs for negative bias parameter values. Then, in Sec.~\ref{sec:06} we consider positive bias parameters where we observe the emergence of multistability, i.e., several simultaneously stable stationary state fixed points, and study the properties of these multiple stable fixed points using the concept of basin stability. In Sec.~\ref{sec:07} we conclude with a discussion of our results and, as this work opens up a new direction of research in nonlinear dynamics and complex networks, some important questions for future work to address.

\section{Nonlinear Random Walks}\label{sec:02}

We begin by describing the nonlinear random walk on a network and presenting some mathematical preliminaries that will be used in our analysis. In particular, in the setting of a nonlinear random walk the transition matrix to depend on the current probability vector explicitly, i.e., $\Pi(\bm{p})$. The system dynamics may then be written as
\begin{align}
p_i(t+1)=\sum_{j=1}^N\pi_{ij}(\bm{p}(t))p_j(t),\hskip2ex\text{or in vector form,}\hskip2ex\bm{p}(t+1)=\Pi(\bm{p}(t))\bm{p}(t).\label{eq:02:01}
\end{align}
These dynamics must map $N$-dimensional probability vectors to $N$-dimensional probability vectors, that is, the vectors $\bm{p}(t)$ must be elements of $\Omega\subset\mathbb{R}^N$, where
\begin{align}
\Omega=\left\{(p_1,p_2,\dots,p_N)^T\in\mathbb{R}^N\bigg|p_i\ge0\text{ for all }i=1,2,\dots,N\text{, and }\sum_{i=1}^Np_i=1\right\}.\label{eq:02:02}
\end{align}
To ensure that this mapping holds we require that $\Pi(\bm{p})$ is column stochastic for all $\bm{p}\in\Omega$, i.e., $\pi_{ij}(\bm{p})\ge0$ for all $i,j$ and $\sum_{i=1}^N\pi_{ij}(\bm{p})=1$ for all $j$. Assuming that the transition probabilities $\pi_{ij}$ are determined by the family of non-negative functions $f_{ij}:\mathbb{R}^N\to[0,\infty)$, we have that
\begin{align}
\pi_{ij}(\bm{p})=\frac{a_{ij}f_{ij}(\bm{p})}{\sum_{l=1}^Na_{lj}f_{lj}(\bm{p})}.\label{eq:02:03}
\end{align}
In this work we will focus our attention on a specific choice of the family of functions $f$ chosen from a family of exponential functions of the current probabilities. In particular, we define the nonlinear random walk below based on using exponential functions to define transition probabilities.

\begin{definition}[Nonlinear Random Walks]\label{def:Nonlinear}
Let $G$ a network of size $N$ with adjacency matrix $A$. The nonlinear, discrete-time random walk on $G$ with parameter $\alpha\in\mathbb{R}$ is given by the dynamics in Eq.~(\ref{eq:02:01}) with transition matrix $\Pi(\bm{p})$ whose entries are given by
\begin{align}
\pi_{ij}(\bm{p})=\frac{a_{ij}\text{exp}(\alpha p_i)}{\sum_{l=1}^Na_{lj}\text{exp}(\alpha p_l)}.\label{eq:02:04}
\end{align}
\end{definition}

\begin{remark}\label{rmk:Nonlinear}
The dynamics of the nonlinear random walk described by Eq.~(\ref{eq:02:01}) and Eq.~(\ref{eq:02:04}) depends on the network structure, encoded in the adjacency matrix $A$ and the nonlinearity parameter $\alpha$. To shed light on the effect of $\alpha$, note first that for the choice $\alpha=0$ we recover the unbiased linear random walk given by transition probabilities $\pi_{ij}=a_{ij}/k_j^{\text{out}}$. Then, by tuning $\alpha$ to be positive or negative we direct random walkers towards or away from nodes depending on the current state of the system. Specifically, $\alpha>0$ directs random walkers towards nodes that already have a more random walkers, while $\alpha<0$ directs random walkers towards nodes that have fewer random walkers. Moreover, this specific choice of transition probabilities based on the current system state properly maps vectors in $\Omega$ into $\Omega$. To see this one can easily check that the the vector $\bm{p}(t+1)=\Pi(\bm{p}(t))\bm{p}(t)$ is non-negative and sums to one.

We also point out that, just as in the case of linear random walks, the nonlinear random walk described by Eq.~(\ref{eq:02:01}) and Eq.~(\ref{eq:02:04}) is memory-less, i.e., Markovian. Thus, the biases in the probabilities directing random walkers from one node to another do not build up over time. Instead, the transition probabilities depend only on the current state of the system.
\end{remark}

Next we discuss some properties of network structures that are critical to the dynamics of both linear and nonlinear random walks. In particular, we will be interested in whether a given network $G$ has an adjacency matrix $A$ that is primitive or non-primitive. Moreover, the primitiveness of an adjacency matrix can be connected to it being irreducible and aperiodic.

\begin{definition}(Irreducible, Aperiodic, and Primitive Matrices)\label{def:Primitive}
Let $M$ be an $N\times N$ non-negative matrix~\cite{Meyer2000}. Then:
\begin{itemize}
\item The matrix $M$ is irreducible if, for all pairs $(i,j)$, there exists an integer $m\ge1$ such that $[M^m]_{ij}>0$.
\item If $M$ is irreducible, the period of $M$ is largest common divisor of all numbers $m\ge1$ that satisfy $[M^m]_{ii}>0$ for any index $i$. The matrix $M$ is aperiodic is the period of $M$ is one.
\item If $M$ is irreducible and aperiodic then it is primitive. Equivalently, $M$ is primitive if there exists an integer $n\ge1$ such that $M^n$ is strictly positive~\cite{Meyer2000}.
\end{itemize}
\end{definition}

\begin{remark}
As mentioned above, our main interest here is in whether a network's adjacency is primitive or non-primitive. (While this can be determined by the reducibility and periodicity, it is often more convenient to inspect higher powers of the adjacency matrix.) Recall that the dynamics of a linear random walk has a unique stationary state fixed point if and only is its adjacency matrix is primitive. As we will see in the following section, the primtiveness of a network's adjacency matrix will guarantee the uniqueness of a stable stationary state fixed point in the weakly nonlinear regime of the nonlinear random walk as well.
\end{remark}

\begin{figure}[htbp]
\centering
\includegraphics[width=0.98\textwidth]{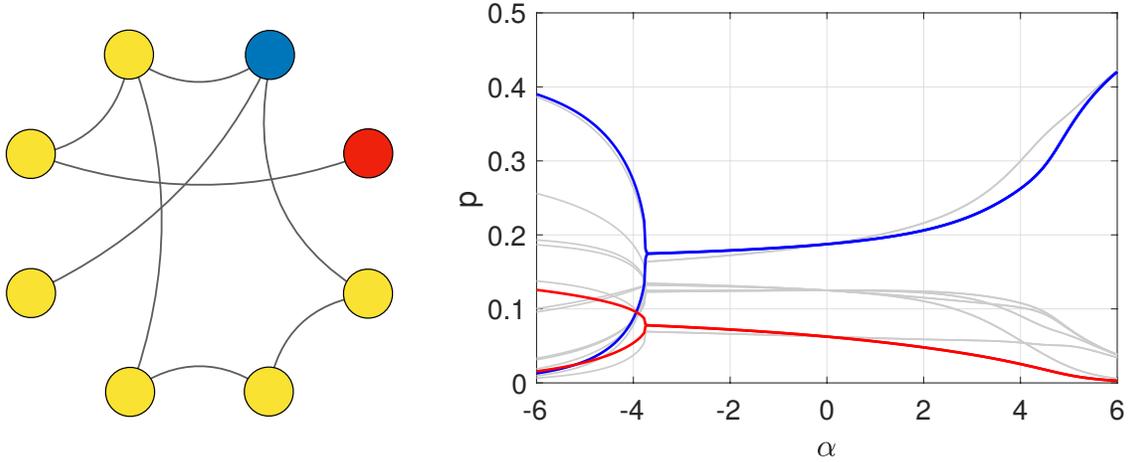}
\caption{{\bf Nonlinear random walks on complex networks: An example}. Left: A network consisting of $N=8$ nodes with $M=8$ links with mean degree $\langle k\rangle=2$. Importantly, the network's adjacency matrix is primitive. Right: the steady-state probabilities $p_i$ (obtained after $5000$ iterations of the nonlinear random walk defined in Eq.~(\ref{eq:02:01}) and Eq.~(\ref{eq:02:04}) as a function of the bias parameter $\alpha$. The probabilities of the nodes colored blue and red are plotted in blue and red, with other probabilities plotted in gray.}\label{fig:01}
\end{figure}

Before proceeding with the analysis of the nonlinear random walk described above, we illustrate the dynamics that emerge for various choices of the bias parameter $\alpha$. In Figure~\ref{fig:01} we show a small network consisting of $N=8$ nodes with $M=8$ links such that the mean degree of the network is $\langle k\rangle=2$. Importantly, this network has a primitive adjacency matrix, with $A^8$ being the lowest power of the matrix that is strictly positive. Next, we simulate the nonlinear random walk dynamics defined by Eq.~(\ref{eq:02:01}) and Eq.~(\ref{eq:02:04}) on this network, plotting the steady-state (obtained after $5000$ iterations for each value of $\alpha$). Specifically, we plot the steady-state probability $p_i$ for each node as a function of $\alpha$, highlighting the probabilities of the nodes colored blue and red with blue and red curves, while all other probabilities are plotted in gray. Note first that for $\alpha=0$ the steady-state probabilities for each node $i$ are proportional to the degree $k_i$. This is because $\alpha=0$ recovers the linear random walk case. Next, for positive values of $\alpha$ (signifying the case where random walkers are preferentially routed towards nodes that already have larger probabilities) the two nodes of larger degree accrue larger probabilities as $\alpha$ is increased while the other nodes all lose random walkers. However, the increase and decrease of the curves $p_i$ display nonlinear behaviors, criss-crossing each other at several locations. Finally, for negative $\alpha$ (signifying the case where random walkers are preferentially routed towards nodes that have smaller probabilities) the trend initially reverses, with all probabilities initially moving closer to one another. However when $\alpha$ becomes sufficiently negative we observe an interesting phenomenon characterized by a period-doubling bifurcation. In particular, the sufficiently negative value of $\alpha$ causes the dynamics of the nonlinear random walk no longer relaxes to a stationary state fixed point solution, but oscillates between two states at each node. Naturally, the onset of oscillating behavior in each node occurs at the same critical value of $\alpha$, however the values between which $p_i$ oscillates varies for each node $i$. For example, the probabilities of the larger degree node illustrated in blue are far more spread out than those of the lower degree node illustrated in red. Finally, we note that as $\alpha$ is further decreased to more negative values no further period-doubling bifurcations are observed in our simulations (not shown).

\section{Stationary State in the Weakly Nonlinear Regime}\label{sec:03}

We begin our analysis of the nonlinear random walk on a complex network by observing in Figure~\ref{fig:01} that for values of the $\alpha$ parameter that are close to zero the stationary state fixed point to which the dynamics converge are close to that of the linear random walk recovered for precisely $\alpha=0$. However, because the linearity of the dynamics is broken for $\alpha\ne0$ the properties that we find in the linear case are no longer guaranteed to hold. Thus, a natural question to pose is, for certain values of $\alpha$ can one guarantee the existence and uniqueness of a stationary state fixed point solution in the nonlinear random walk towards which the dynamics converge, similar to the linear case? In this section we answer this question in the affirmative provided that the network has a primitive adjacency matrix. Specifically, we prove in the analysis below that for a bias parameter that is sufficiently close to zero, i.e., $|\alpha|\ll1$, if the network's adjacency matrix is primitive, then the nonlinear random walk converges to a stationary state fixed point. We therefore refer to the regime where $|\alpha|$ is sufficiently small as the weakly nonlinear regime. Moreover, we will demonstrate that, is in the linear random walk case, if the network's adjacency matrix is non-primitive, then the dynamics tend not to converge to a uniques stationary state fixed point solution.

To prove that the nonlinear random walk on a network with a primitive adjacency matrix converges to a unique stationary state fixed point for sufficiently small $|\alpha|$ we will show that iterating the dynamics forward in time constitutes a contraction mapping. We will then rely on the contraction mapping theorem to complete the proof~\cite{Hunter2001}. In \ref{app:A} we summarize the definition of a contraction mapping and the contraction mapping theorem. To formalize this we require a complete metric space on which the contraction mapping acts. Here, we will consider the space $\Omega$, as defined in Eq.~(\ref{eq:02:02}) equipped with the $\ell^1$, or taxicab metric, defined, for $\bm{p},\bm{q}\in\Omega$, by
\begin{align}
d(\bm{p},\bm{q})=\|\bm{p}-\bm{q}\|_1=\sum_{i=1}^N|p_i-q_i|.\label{eq:03:01}
\end{align}
We will also utilize the matrix norm induced by $\ell^1$, given by, for an $N\times N$ matrix $A$ with entries $[A]_{ij}$,
\begin{align}
\|A\|_1=\sup_{\|x\|_1=1}\|Ax\|_1=\max_{1\le j\le N}\sum_{i=1}^N|a_{ij}|.\label{eq:03:02}
\end{align}
Note that the matrix norm corresponds to the maximum column sum of the matrix after taking the absolute value of each entry. We begin by asserting a simply property about the $\ell^1$ norm of a transition matrix.

\begin{lemma}\label{Lemma:ell1}
Let $M$ be an $N\times N$ transition matrix. Then $\|M\|_1=1$.
\end{lemma}
\begin{proof}
Since $M$ is a transition matrix, all of its entries are non-negative and each column sums up to one. Denoting its entries $[M]_{ij}=m_{ij}$, we have that
\begin{align}
\|M\|_1=\max_{1\le j\le N}\sum_{i=1}^N|m_{ij}|=\max_{1\le j\le N}\sum_{i=1}^Nm_{ij}=\max_{1\le j\le N}1=1,\label{eq:03:03}
\end{align}
thus completing the proof.
\end{proof}

We now turn our attention to the nonlinear transition matrix $\Pi(\bm{p})$ whose entries are given in Eq.~(\ref{eq:02:04}). Our focus here is on the case of weak nonlinearity, i.e., $|\alpha|\ll1$. In this regime we may treat $\Pi(\bm{p})$ using a series expansion where high-order terms involving powers of $\alpha$ may be neglected. In particular, the entries of $\Pi(\bm{p})$ may be written as
\begin{align}
\pi_{ij}(\bm{p})&=\frac{a_{ij}}{k_{j}}+\alpha\frac{a_{ij}}{k_{j}}\left(p_i-\frac{1}{k_j}\sum_{l=1}^Na_{lj}p_l\right)+\mathcal{O}(\alpha^2)\label{eq:03:04}\\
&=\pi_{ij}^{(0)}+\alpha~\pi_{ij}^{(1)}(\bm{p})+\mathcal{O}(\alpha^2),\nonumber
\end{align}
where $\pi_{ij}^{(0)}$ and $\pi_{ij}^{(1)}(\bm{p})$ represent the zero-order term and first-order correction to the transition probability $\pi_{ij}(\bm{p})$ and the respective entries of the matrices $\Pi^{(0)}$ and $\Pi^{(1)}(\bm{p})$. We note here that these two matrices have important properties. First, $\Pi^{(0)}$ is precisely the transition matrix for the (unbiased) linear random walk on the network with adjacency matrix $A$, and therefore $\Pi^{(0)}$ is itself a transition matrix. Second, the columns of the correction matrix $\Pi^{(1)}(\bm{p})$ sum to zero:
\begin{align}
\sum_{i=1}^N\pi_{ij}^{(1)}(\bm{p})&=\sum_{i=1}^N\frac{a_{ij}}{k_{j}}\left(p_i-\frac{1}{k_j}\sum_{l=1}^Na_{lj}p_l\right)\label{eq:03:05}\\
&=\sum_{i=1}^N\frac{a_{ij}p_i}{k_{j}}-\sum_{i=1}^N\frac{a_{ij}}{(k_j)^2}\sum_{l=1}^Na_{lj}p_l(t)\nonumber\\
&=\sum_{i=1}^N\frac{a_{ij}p_i}{k_{j}}-\frac{1}{k_j}\sum_{l=1}^Na_{lj}p_l(t)=0.\nonumber
\end{align}
Moreover, the non-zero entries of $\Pi^{(1)}(\bm{p})$ coincide with the positive entries of $\Pi^{(0)}$, so for sufficiently small $|\alpha|$ we can ensure that the entries of the linearized matrix $\Pi^{(0)}+\alpha\Pi^{(1)}(\bm{p})$ are non-negative and the columns sum to one, so that it is a (column) stochastic matrix. Moreover, $|\alpha|$ we be chosen small enough to neglect terms of order $\mathcal{O}(\alpha^2)$ and higher, so that we may work with the linearized transition matrix. 

Turning our attention back to the nonlinear random walk, we note that in the case of the linear random walk the existence and uniqueness of the stationary state follows from the property of the adjacency matrix $A$ (and therefore the transition matrix $\Pi$) being primitive, that is, for some integer $n\ge1$ the matrix $A^n$ is strictly positive. Here we maintain this requirement on $A$ for the nonlinear random walk to guarantee a unique stationary state. We then define the the matrix that maps the dynamics forward $n$ steps, and prove some important properties about this operator.

\begin{definition}
Consider a network $G$ with primitive adjacency matrix $A$. Denote $n\ge1$ as the smallest possible integer satisfying $A^n$ is strictly positive. Then define $Q_n(\bm{p})$ as the matrix that maps the dynamics of the nonlinear random walk forward $n$ steps, so that $\bm{p}(t+n)=Q_n(\bm{p}(t))\bm{p}(t)$:
\begin{align}
Q_n(\bm{p}(t))=\Pi(\bm{p}(t+n-1))\Pi(\bm{p}(t+n-2))\cdots\Pi(\bm{p}(t+1))\Pi(\bm{p}(t))=\prod_{m=0}^{n-1}\Pi(\bm{p}(t+m)).\label{eq:03:06}
\end{align}
\end{definition}

We are currently interested in the properties of $Q_n(\bm{p})$ in the small nonlinearity regime. Note that, using the expansion of $\Pi(\bm{p})$ in Eq.~(\ref{eq:03:04}), we may write
\begin{align}
Q_n(\bm{p}(t))&=\left(\prod_{m=0}^{n-1}\left(\Pi^{(0)}+\alpha\Pi^{(1)}(\bm{p}(t+m))+\mathcal{O}(\alpha^2)\right)\right)\label{eq:03:07}\\
&=\left(\Pi^{(0)}\right)^n+\alpha\left(\Pi^{(0)}\right)^{n-1}\left(\sum_{m=0}^{n-1}\Pi^{(1)}(\bm{p}(t+m))\right)+\mathcal{O}(\alpha^2).\nonumber
\end{align}
Next we prove some important properties of the terms in Eq.~(\ref{eq:03:07}). First, the leading order term is strictly positive.
\begin{lemma}\label{Lemma:Positive}
Consider a network $G$ with primitive adjacency matrix $A$. Denote $n\ge1$ as the smallest possible integer satisfying $A^n$ is strictly positive. Then $(\Pi^{(0)})^n$, whose entries are denoted in Eq.~(\ref{eq:03:04}), is strictly positive.
\end{lemma}
\begin{proof}
Recall that the entries of $\Pi^{(0)}$ are given by $\pi_{ij}^{(0)}=a_{ij}/k_j$, so that
\begin{align}
\pi_{ij}^{(0)}=\frac{a_{ij}}{k_j}\ge \frac{a_{ij}}{k_{\text{max}}},\label{eq:03:08}
\end{align}
where $k_{\text{max}}=\max_{j}k_j$ is the maximum nodal out-degree and $k_{\text{max}}^{\text{out}}>0$. Then, taking the $n^{\text{th}}$ power of $\Pi^{(0)}$ we have
\begin{align}
(\Pi^{(0)})^n\ge\left(\frac{1}{k_{\text{max}}}A\right)^n=\frac{1}{(k_{\text{max}})^n}A^n,\label{eq:03:09}
\end{align}
which is strictly positive, completing the proof.
\end{proof}
We also prove that the $\ell^1$ norm of the summation term in Eq.~(\ref{eq:03:07}) is bounded.
\begin{lemma}\label{Lemma:Bounded}
Consider a network $G$ with adjacency matrix $A$. Then
\begin{align}
\left\|\sum_{m=0}^{n-1}\Pi^{(1)}(\bm{p}(t+m))\right\|_1\le 2n,\label{eq:03:10}
\end{align}
where the entries of $\Pi^{(1)}(\bm{p}(t+m))$ are denoted in Eq.~(\ref{eq:03:04}).
\end{lemma}
\begin{proof}
Applying the triangle inequality, interpreting the $\ell^1$ norm, and using that $0\le p_i\le1$, we have that
\begin{align}
\left\|\sum_{m=0}^{n-1}\Pi^{(1)}(\bm{p}(t+m))\right\|_1&\le\sum_{m=0}^{n-1}\left\|\Pi^{(1)}(\bm{p}(t+m))\right\|_1\label{eq:03:11}\\
&=\sum_{m=0}^{n-1}\left(\max_{1\le j\le N}\sum_{i=1}^N\left|\frac{a_{ij} p_i}{k_j}-\frac{a_{ij}}{(k_j)^2}\sum_{l=1}^Na_{lj}p_l\right|\right)\nonumber\\
&\le\sum_{m=0}^{n-1}\left(\max_{1\le j\le N}\frac{\sum_{i=1}^Na_{ij}p_i}{k_j}+\frac{\sum_{i=1}^Na_{ij}\sum_{l=1}^Na_{lj}p_l}{(k_j)^2}\right)\nonumber\\
&=2\sum_{m=0}^{n-1}\max_{1\le j\le N} \left(\frac{\sum_{i=1}^Na_{ij}p_i}{k_j}\right)\nonumber\\
&\le2\sum_{m=0}^{n-1}\max_{1\le j\le N} \left(\frac{\sum_{i=1}^Na_{ij}}{k_j}\right)\nonumber\\
&=2\sum_{m=0}^{n-1}\max_{1\le j\le N}1=2n,\nonumber
\end{align}
which yields the desired result.
\end{proof}
With these results we can show show that, for sufficiently small nonlinearity, the operation of multiplying by the matrix $Q_n(\bm{p})$ is a contraction mapping on the space $\Omega$ using the $\ell^1$ distance.

\begin{lemma}\label{Lemma:Contraction}
Consider a network $G$ with primitive adjacency matrix $A$ and the nonlinear random walk defined by Eq.~(\ref{eq:02:01}) and Eq.~(\ref{eq:02:04}). Denote $n\ge1$ as the smallest possible integer satisfying $A^n$ is strictly positive. Then for sufficiently small (but non-zero) nonlinearity, i.e., small $|\alpha|$, the operator $Q_n(\bm{p})$, defined in Eq.~(\ref{eq:03:06}), is a contraction mapping on $\Omega$, defined in Eq.~(\ref{eq:02:02}), using the $\ell^1$ distance defined in Eq.~(\ref{eq:03:01}).
\end{lemma}

\begin{proof}
We consider the $\ell^1$ distance between the action of $Q_n$ on two vectors $\bm{p}(t)$ and $\bm{q}(t)$, both in $\Omega$. In particular, we aim to show that $\|Q_n(\bm{p}(t))\bm{p}(t)-Q_n(\bm{q}(t))\bm{q}(t)\|_1\le c\|\bm{p}(t)-\bm{q}(t)\|_1$ for some constant $c$, where $0\le c<1$. From Eq.~(\ref{eq:03:07}) we have that
\begin{align}
\|Q_n&(\bm{p}(t))\bm{p}(t)-Q_n(\bm{q}(t))\bm{q}(t)\|_1=\left\|\left(\Pi^{(0)}\right)^n(\bm{p}(t)-\bm{q}(t))\right.\label{eq:03:12}\\
&\left.+\alpha\left(\Pi^{(0)}\right)^{n-1}\left(\left(\sum_{j=0}^{n-1}\Pi^{(1)}(\bm{p}(t+j))\right)\bm{p}(t)-\left(\sum_{j=0}^{n-1}\Pi^{(1)}(\bm{q}(t+j))\right)\bm{q}(t)\right)+\mathcal{O}(\alpha^2)\right\|_1,\nonumber
\end{align}
and applying the trsiangle inequality yields
\begin{align}
\|Q_n&(\bm{p}(t))\bm{p}(t)-Q_n(\bm{q}(t))\bm{q}(t)\|_1\le\overbrace{\left\|\left(\Pi^{(0)}\right)^n(\bm{p}(t)-\bm{q}(t))\right\|_1}^{\text{term A}}\label{eq:03:13}\\
&+\underbrace{\alpha\left\|\left(\Pi^{(0)}\right)^{n-1}\left(\left(\sum_{j=0}^{n-1}\Pi^{(1)}(\bm{p}(t+j))\right)\bm{p}(t)-\left(\sum_{j=0}^{n-1}\Pi^{(1)}(\bm{q}(t+j))\right)\bm{q}(t)\right)\right\|_1}_{\text{term B}}+\underbrace{\mathcal{O}(\alpha^2)}_{\text{term C}},\nonumber
\end{align}
so we can now treat terms A, B, and C separately. 

First we treat term A in Eq.~(\ref{eq:03:13}). Note that from \ref{Lemma:Positive} we have that $(\Pi^{(0)})^n$ is strictly positive. Therefore, we may choose some $\epsilon$ satisfying $0<\epsilon<N^{-1}$ to construct a matrix $M$ whose entries are given by $[M_{ij}]=m_{ij}=(\pi_{ij}^{(0)n}-\epsilon)/(1-N\epsilon)$ are non-negative, where $\pi_{ij}^{(0)n}=[(\Pi^{(0)})^n]_{ij}$. Note that $M$ is a transition matrix because its entries are non-negative and sum to one:
\begin{align}
\sum_{i=1}^Nm_{ij}=\frac{1}{1-N\epsilon}\sum_{i=1}^N(\pi_{ij}^{(0)n}-\epsilon)=\frac{1}{1-N\epsilon}-\frac{N\epsilon}{1-N\epsilon}=1.\label{eq:03:14}
\end{align}
Inspecting term A in Eq.~(\ref{eq:03:13}) and substituting in the entries of $P$ for those of $M$ allows us to bound it as follows:
\begin{align}
\left\|\left(\Pi^{(0)}\right)^n(\bm{p}(t)-\bm{q}(t))\right\|_1&=\sum_{i=1}^N\left|\left[\left(\Pi^{(0)}\right)^n\bm{p}(t)\right]_i-\left[\left(\Pi^{(0)}\right)^n\bm{q}(t)\right]_i\right|\label{eq:03:15}\\
&= \sum_{i=1}^N\left|\sum_{j=1}^N\pi_{ij}^{(0)n}q_j(t)-\pi_{ij}^{(0)n}q_j(t)\right|\nonumber\\
&= \sum_{i=1}^N\left|\sum_{j=1}^N(1-N\epsilon)m_{ij}(p_j(t)-q_j(t))\right|\nonumber\\
&\le(1-N\epsilon)\sum_{i=1}^N\sum_{j=1}^Nm_{ij}|p_j(t)-q_j(t)|\nonumber\\
&\le(1-N\epsilon)\sum_{j=1}^N|p_j(t)-q_j(t)|\sum_{i=1}^Nm_{ij}\nonumber\\
&=(1-N\epsilon)\sum_{j=1}^N|p_j(t)-q_j(t)|=(1-N\epsilon)\|\bm{p}(t)-\bm{q}(t)\|_1.\nonumber
\end{align}

We now shift our attention to term B in Eq.~(\ref{eq:03:11}). Using the properties of matrix norms, specifically that $\|A\bm{x}\|_1\le\|A\|_1\|\bm{x}\|_1$ for $N\times N$ matrices $A$ and $N$-dimensional vectors $\bm{x}$, we first write
\begin{align}
\alpha&\left\|\left(\Pi^{(0)}\right)^{n-1}\left(\left(\sum_{j=0}^{n-1}\Pi^{(1)}(\bm{p}(t+j))\right)\bm{p}(t)-\left(\sum_{j=0}^{n-1}\Pi^{(1)}(\bm{q}(t+j))\right)\bm{q}(t)\right)\right\|_1\label{eq:03:16}\\
&\hskip22ex\le\alpha\left\|\left(\Pi^{(0)}\right)^{n-1}\right\|_1\max_{\bm{x}(t)\in\Omega}\left\|\left(\sum_{m=0}^{n=1}\Pi^{(1)}(\bm{x}(t))\right)\right\|_1\left\|\bm{p}(t)-\bm{q}(t)\right\|_1.\nonumber
\end{align}
Next, because $(\Pi^{(0)})^n$ is a transition matrix, it follows from \ref{Lemma:ell1} that $\|(\Pi^{(0)})^{n-1}\|_1=1$, and \ref{Lemma:Bounded} implies that 
\begin{align}
\max_{\bm{x}(t)\in\Omega}\left\|\left(\sum_{m=0}^{n=1}\Pi^{(1)}(\bm{x}(t))\right)\right\|_1\le 2n.\label{eq:03:17}
\end{align}
Thus, we have that
\begin{align}
\alpha&\left\|\left(\Pi^{(0)}\right)^{n-1}\left(\left(\sum_{j=0}^{n-1}\Pi^{(1)}(\bm{p}(t+j))\right)\bm{p}(t)-\left(\sum_{j=0}^{n-1}\Pi^{(1)}(\bm{q}(t+j))\right)\bm{q}(t)\right)\right\|_1\label{eq:03:18}\\
&\hskip54ex\le \alpha(2n)\left\|\bm{p}(t)-\bm{q}(t)\right\|_1.\nonumber
\end{align}

Having treated terms A and B in Eq.~(\ref{eq:03:13}), we now insert the inequalities in Eq.~(\ref{eq:03:15}) and Eq.~(\ref{eq:03:18}) into Eq.~(\ref{eq:03:13}) to obtain
\begin{align}
\left\|Q_n(\bm{p}(t))\bm{p}(t)-Q_n(\bm{q}(t))\bm{q}(t)\right\|_1\le\left( (1-N\epsilon) +\alpha(2n)\right)\|\bm{p}(t)-\bm{q}(t)\|_1+\mathcal{O}(\alpha^2).\label{eq:03:19}
\end{align}
Inspecting Eq.~(\ref{eq:03:19}), we may choose $|\alpha|<N\epsilon/(2n)$ to obtain
\begin{align}
\left\|Q_n(\bm{p}(t))\bm{p}(t)-Q_n(\bm{q}(t))\bm{q}(t)\right\|_1\le c\|\bm{p}(t)-\bm{q}(t)\|_1+\mathcal{O}(\alpha^2).\label{eq:03:20}
\end{align}
with $0\le c<1$, where $c=1-N\epsilon+2\alpha n$, thus proving the the linearization of $Q_n$ is a contraction on $\Omega$ using the $\ell^1$ distance. Moreover, since we are considering the small nonlinearity regime, $|\alpha|$ may be reduced further to ensure that the terms of higher order, i.e., $\mathcal{O}(\alpha^2)$ may be neglected, thus completing the proof.
\end{proof}

Having asserted that, for a network with a primitive adjacency matrix, the forward mapping of the nonlinear random walk an appropriate number of iterations constitutes a contraction mapping, we prove that for a sufficiently small bias parameter the nonlinear random walk converges to a unique stationary state fixed point.

\begin{theorem}[Existence and Uniqueness of the Stable Stationary State Fixed Point]\label{Theorem:E&U}
Consider a network $G$ with primitive adjacency matrix $A$ and the nonlinear random walk defined by Eq.~(\ref{eq:02:01}) and Eq.~(\ref{eq:02:04}). Then for sufficiently small values of the nonlinearity parameter $|\alpha|$ there exists a unique stationary state fixed point $\bm{p}^*$ satisfying
\begin{align}
\bm{p^*}=\Pi(\bm{p^*})\bm{p^*}.\label{eq:03:21}
\end{align}
Moreover, $\bm{p}^*$ is globally attracting, with
\begin{align}
\lim_{t\to\infty}\bm{p}(t)=\bm{p}^*\label{eq:03:22}
\end{align}
for any initial $\bm{p}(0)\in\Omega$.
\end{theorem}

\begin{proof}
Let $n\ge1$ be the smallest integer such that $A^n$ is strictly positive. By \ref{Lemma:Contraction}, the matrix operator $Q_n$ as defined in Eq.~(\ref{eq:03:06}) for sufficiently small $|\alpha|$ is a contraction on $\Omega$ using the $\ell^1$ norm. By the contraction mapping theorem (\ref{Theorem:CMT}) there exists a unique fixed point $\bm{p}^*$ of the mapping $Q_n:\Omega\to\Omega$ satisfying $\bm{p}^*=Q_n(\bm{p}^*)\bm{p^*}$. Moreover, this fixed point is attracting for any initial condition. However, note that $Q_n$ maps the dynamics forward $n$ time steps, so this fixed point may in principle be a periodic orbit of the nonlinear random walk dynamics with a period between $1$ and $n$. Specifically, this implies that
\begin{align}
\lim_{m\to\infty}\bm{p}(mn)=\bm{p^*},\label{eq:03:23}
\end{align}
for any initial condition $\bm{p}(0)\in\Omega$. However, the theorem may also be applied to the initial conditions $\bm{p}(1)=\Pi(\bm{p}(0))\bm{p}(0)$, $\bm{p}(2)=\Pi(\bm{p}(1))\bm{p}(1)$, $\cdots$, and $\bm{p}(n-1)=\Pi(\bm{p}(n-2))\bm{p}(n-2)$. Thus, we have that 
\begin{align}
\lim_{m\to\infty}\bm{p}(k+mn)=\bm{p^*},\label{eq:03:24}
\end{align}
for all any $k=0,1,\dots,n-1$. Thus, every iterate of the period-$n$ fixed point is equal to $\bm{p}^*$. Therefore, the period of $\bm{p}^*$ in the nonlinear random walk dynamics is period one and a true fixed point that is globally attracting, completing the proof of the theorem.
\end{proof}

\begin{remark}
This result asserts that, provided that a network's adjacency matrix is primitive, the biasing parameter $\alpha$ can be tuned appropriately (but remaining non-zero) to ensure that the nonlinear random walk dynamics converge to a unique stationary state fixed point analogous to the linear case. While the proof above assumes that $|\alpha|\ll1$, in practice we find that often $|\alpha|$ take modestly larger values with the dynamics converging to a unique stationary state fixed point in the long run.
\end{remark}

\begin{figure}[htbp]
\centering
\includegraphics[width=0.98\textwidth]{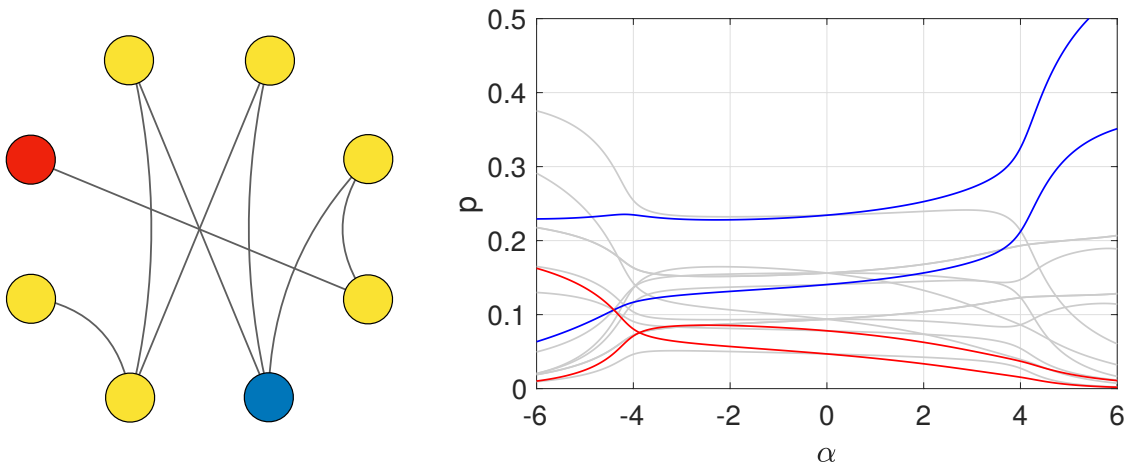}
\caption{{\bf The non-primitive case.}. Left: A network consisting of $N=8$ nodes with $M=8$ links with mean degree $\langle k\rangle=2$. Importantly, the network's adjacency matrix is non-primitive. Right: the steady-state probabilities $p_i$ (obtained after $5000$ iterations of the nonlinear random walk defined in Eq.~(\ref{eq:02:01}) and Eq.~(\ref{eq:02:04}) as a function of the bias parameter $\alpha$. The probabilities of the nodes colored blue and red are plotted in blue and red, with other probabilities plotted in gray.}\label{fig:02}
\end{figure}

As in the case of the linear random walk, we find that the nonlinear random walk on a network whose adjacency matrix is non-primitive does not have a unique stationary state fixed point to which the dynamics converge. To illustrate this this, consider the network illustrated in Figure~\ref{fig:02}, which has a non-primitive adjacency matrix. (To see this, note that this is a bipartite network.) Like the network illustrated in Figure~\ref{fig:01}, the network consists of $N=8$ nodes, $M=8$ links, and has a mean degree of $\langle k\rangle=2$. We plot as a function of $\alpha$ the steady-state probabilities $p_i$ (obtained after $5000$ iterations at each value of $\alpha$) in blue and red for the nodes colored blue and red, with other probabilities plotted in gray. Specifically, the branches obtained here are obtained by using an initial condition of uniform probabilities, i.e., $\bm{p}(0)=[N^{-1},\dots,N^{-1}]^T$ at $\alpha=0$ and numerically continuing the solution by increasing and decreasing $\alpha$ by small increments. Even at precisely $\alpha=0$ the long-term dynamics oscillate in a period-two pattern, which is also the case for $\alpha\ne0$, illustrating that the long-term dynamics tend to be periodic. We note also that, depending on the network structure, these periodic oscillations may be of higher periodicity, and for each non-primitive network and value of $\alpha$ different periodic solutions may be obtained by choosing different initial conditions (not shown).

\section{Asymptotic Analysis of the Stationary State}\label{sec:04}

Having asserted the existence and uniqueness of a stable stationary state fixed point for the nonlinear random walk dynamics in the weakly nonlinear regime on a network with primitive adjacency matrix, we now turn our attention to a question of a more quantitative nature. In particular, we seek an approximation for the stationary state fixed point in the weakly nonlinear regime. Because the weakly nonlinear regime is characterized by sufficiently small $|\alpha|$, we proceed perturbatively, taking $\alpha$ as a small parameter and seek an asymptotic series solution of which we compute the first few terms presented below.

\begin{theorem}[Asymptotic Series for the Stationary State]\label{Theorem:Asymptotics}
Consider a network $G$ with primitive adjacency matrix $A$ and the nonlinear random walk defined by Eq.~(\ref{eq:02:01}) and Eq.~(\ref{eq:02:04}). Then for sufficiently small values of the nonlinearity parameter $|\alpha|$ the stationary state solution is given by
\begin{align}
\bm{p}^*=\bm{p^{(0)}}+\alpha\bm{p}^{(1)}+\alpha^2\bm{p}^{(2)}+\mathcal{O}(\alpha^3),\label{eq:04:01}
\end{align}
where the leading-order term is given by
\begin{align}
\bm{p}^{(0)}=\frac{\bm{k}}{\sum_{l=1}^Nk_l},\label{eq:04:02}
\end{align}
and the correction terms $\bm{p}^{(1)}$ and $\bm{p}^{(2)}$ are the zero-sum solution to
\begin{align}
(I-AD^{-1})\bm{p}^{(1)}=\frac{(DA-AD^{-1}A^TD)\bm{1}}{\left(\sum_{l=1}^Nk_l\right)^2},\label{eq:04:03}
\end{align}
and
\begin{align}
(I-AD^{-1})\bm{p}^{(2)}&=\frac{\left(DAD^{-1}P^{(1)}-AD^{-2}P^{(1)}A^TD\right)\bm{1}}{\sum_{l=1}^Nk_l} + \frac{\left(P^{(1)}A-AD^{-1}A^TP^{(1)}\right)\bm{1}}{\sum_{l=1}^Nk_l}\label{eq:04:04}\\&+\frac{\left(D^2A - AD^{-1}AD^2\right)\bm{1}}{\left(\sum_{l=1}^Nk_l\right)^3}+\frac{\left(AD^{-2}T - DAD^{-1}A^TD\right)\bm{1}}{\left(\sum_{l=1}^Nk_l\right)^3},\nonumber
\end{align}
where $\bm{1}=(1,1,\dots,1)^T$, $T$ is the diagonal matrix with entries $[T]_{jj} = \left(\sum_l a_{lj}k_l\right)^2$, and $P^{(1)}$ is the diagonal matrix with entries $[P^{(1)}]_{jj} = p_j^{(1)}$.
\end{theorem}

\begin{proof}
We seek a stationary state solution to the dynamics of Eq.~(\ref{eq:02:01}) and Eq.~(\ref{eq:02:04}) of the form $\bm{p}^*=\Pi(\bm{p}^*)\bm{p^*}$, whose $i^{\text{th}}$ component is given by
\begin{align}
p_i^*=\sum_{j=1}^N\pi_{ij(\bm{p}^*)}p_j.\label{eq:04:05}
\end{align}
Moreover, we consider series solutions of the form in Eq.~(\ref{eq:04:01}), whose $i^{\text{th}}$ component is given by
\begin{align}
p_i^*=p_i^{(0)}+\alpha p_i^{(1)}+\alpha^2 p_i^{(2)}+\mathcal{O}(\alpha^3).\label{eq:04:06}
\end{align}
Next we expand the terms of the transition matrix $\Pi(\bm{p}^*)$, using that $\text{exp}(\alpha p)=1+\alpha p+\alpha^2 p^2+\mathcal{O}(\alpha^3)$, which yields
\begin{align}
\pi_{ij}(\bm{p})&=\frac{a_{ij}e^{\alpha p_i}}{\sum_{l=1}^N a_{lj} e^{\alpha p_l}}=\frac{a_{ij}}{k_j}+\alpha\frac{a_{ij}}{k_j}\left(p_i-\frac{1}{k_j}\sum_{l=1}^Na_{lj}p_l\right)\label{eq:04:07}\\&\hskip4ex+\alpha^2\frac{a_{ij}}{k_j}\left(p_i^2-\frac{1}{k_j}\sum_{l=1}^Na_{lj}p_l^2+\frac{1}{k_j^2}\left(\sum_{l=1}^Na_{lj}p_l\right)^2-\frac{p_i}{k_j}\sum_{l=1}^Na_{lj}p_l\right)+\mathcal{O}(\alpha^3).\nonumber
\end{align}
We then insert Eq.~(\ref{eq:04:06}) into Eq.~(\ref{eq:04:07}) and subsequently Eq.~(\ref{eq:04:06}) and Eq.~(\ref{eq:04:07}) into Eq.~(\ref{eq:04:05}) which yields a complicated expression then collect terms at different powers of $\alpha$. Starting with the leading order terms, i.e., terms of order $\mathcal{O}(\alpha^0)$, yields the expression
\begin{align}
p_i^{(0)}=\sum_{j=1}^N\frac{a_{ij}}{k_j}p_j^{(0)},\label{eq:04:08}
\end{align}
which is solved by any $p_i\propto k_i$. Normalizing to ensure that $\sum_{l}p_l=1$ yields $p_i=k_i/\sum_lk_l$, or in vector form, the expression given in Eq.~(\ref{eq:04:02}).

Before proceeding to the linear correction term, we reexamine our proposed solution, i.e., Eq.~(\ref{eq:04:01}) [or Eq.~(\ref{eq:04:06})]. Since the leading order term $\bm{p}^{(0)}$ is a probability vector, i.e., sums to one, to ensure that the asymptotic series remains a probability vector we constrain the higher order correction term, i.e., $\bm{p}^{(1)}$, $\bm{p}^{(2)}$, etc., to sum to zero. Then, checking terms of order $\mathcal{O}(\alpha^1)$ yields
\begin{align}
p_i^{(1)}=\sum_{j=1}^N\frac{a_{ij}}{k_j}p_j^{(1)}+\frac{1}{\left(\sum_{l=1}^Nk_l\right)^2}\sum_{j=1}^Na_{ij}\left(k_i-\frac{1}{k_j}\sum_{l=1}^Na_{lj}k_l\right).\label{eq:04:09}
\end{align}
Rearranging to put entries of $\bm{p}^{(1)}$ on the left hand side and expressing Eq.~(\ref{eq:04:09}) in vector form yields the desired express given in Eq.~(\ref{eq:04:03}).

Finally, checking terms of order $\mathcal{O}(\alpha^2)$ yields
\begin{align}
p_i^{(2)}=\sum_{j=1}^N\frac{a_{ij}}{k_j}p_j^{(2)}&+\frac{1}{\sum_{l=1}^Nk_l}\sum_{j=1}^N\frac{a_{ij}}{k_j}q_j\left(k_i-\frac{1}{k_j}\sum_{l=1}^Na_{lj}k_l\right)\label{eq:04:10}\\
&+\frac{1}{\sum_{l=1}^Nk_l}\sum_{j=1}^Na_{ij}\left(p_i^{(1)}-\frac{1}{k_j}\sum_{l=1}^Na_{lj}p_l^{(1)}\right)\nonumber\\
&+\frac{1}{\left(\sum_{l=1}^Nk_l\right)^3}\sum_{j=1}^Na_{ij}\left(k_i^2-\frac{1}{k_j}\sum_{l=1}^Na_{lj}k_l^2\right)\nonumber\\
&+\frac{1}{\left(\sum_{l=1}^Nk_l\right)^3}\sum_{j=1}^N\frac{a_{ij}}{k_j}\left(\frac{1}{k_j}\left(\sum_{l=1}^Na_{lj}k_l\right)^2-k_i\sum_{l=1}^Na_{lj}k_l\right).\nonumber
\end{align}
Rearranging to put entries of $\bm{p}^{(2)}$ on the left hand side and expressing Eq.~(\ref{eq:04:10}) in vector form, where $P^{(1)}=\text{diag}(p_1^{(1)},p_2^{(1)},\dots,p_N^{(1)})$ and $T=\text{diag}((\sum_{l}a_{l1}k_l)^2,(\sum_{l}a_{l2}k_l)^2,\dots,\sum_{l}a_{lN}k_l)^2)$ yields the desired express given in Eq.~(\ref{eq:04:04}), completing the proof.
\end{proof}

\begin{remark}
The computation of the approximation of the stationary state fixed point solution in Eq.~(\ref{eq:04:01}), in particular the correction terms $\bm{p}^{(1)}$ and $\bm{p}^{(2)}$ given in Eq.~(\ref{eq:04:03}) and Eq.~(\ref{eq:04:04}) deserves some additional discussion. In particular, we note that the matrix $AD^{-1}$ is the transition matrix $\Pi$ for the unbiased linear random walk for the network with adjacency matrix $A$. Thus, $AD^{-1}$ has at least one eigenvalue $\lambda=1$ (precisely one such eigenvalue if $A$ is primitive), so that the matrix $I-AD^{-1}$ that appears on the left hand side of Eq.~(\ref{eq:04:03}) and Eq.~(\ref{eq:04:04}) has an eigenvalue of $\lambda=0$ and is singular. Thus, no unique solution exists to either Eq.~(\ref{eq:04:03}) or Eq.~(\ref{eq:04:04}). However, the requirement that $\bm{p}^{(1)}$ and $\bm{p}^{(2)}$ sum to zero adds an additional constraint that allows us to find the correct solution to our problem. 

Moreover, the singular nature of the matrix $I-AD^{-1}$ complicates the computation of $\bm{p}^{(1)}$ and $\bm{p}^{(2)}$, but can be overcome by considering the following fixed point iteration schemes. Specifically, we find that in practice it is convenient to solve Eq.~(\ref{eq:04:03}) and Eq.~(\ref{eq:04:04}) using
\begin{align}
\bm{p}_{n+1}^{(1)}=AD^{-1}\bm{p}_n^{(1)}+\frac{(DA-AD^{-1}A^TD)\bm{1}}{\left(\sum_{l=1}^Nk_l\right)^2},\label{eq:04:11}
\end{align}
and
\begin{align}
\bm{p}_{n+1}^{(2)}&=AD^{-1}\bm{p}_n^{(2)}+\frac{\left(DAD^{-1}P^{(1)}-AD^{-2}P^{(1)}A^TD\right)\bm{1}}{\sum_{l=1}^Nk_l} + \frac{\left(P^{(1)}A-AD^{-1}A^TP^{(1)}\right)\bm{1}}{\sum_{l=1}^Nk_l}\label{eq:04:12}\\&+\frac{\left(D^2A - AD^{-1}AD^2\right)\bm{1}}{\left(\sum_{l=1}^Nk_l\right)^3}+\frac{\left(AD^{-2}T - DAD^{-1}A^TD\right)\bm{1}}{\left(\sum_{l=1}^Nk_l\right)^3}.\nonumber
\end{align}
We note first that both iterations converge to a solution because the eigenvalues of $AD^{-1}$ satisfy $|\lambda|<1$ except for precisely one eigenvalue $\lambda=1$ (which gives rise to the non-uniqueness of solutions). Moreover, it can be checked that the non-homogeneous terms on the right hand side of Eq.~(\ref{eq:04:11}) and Eq.~(\ref{eq:04:12}) sum to zero, and multiplication by $AD^{-1}$ preserved vector sums, so provided that the initial conditions $\bm{p}_0^{(1)}$ and $\bm{p}_0^{(2)}$ sum to zero (for instance, we use the zero vector), the iterations converge to the desired solutions.
\end{remark}

\begin{figure}[htbp]
\centering
\includegraphics[width=0.49\textwidth]{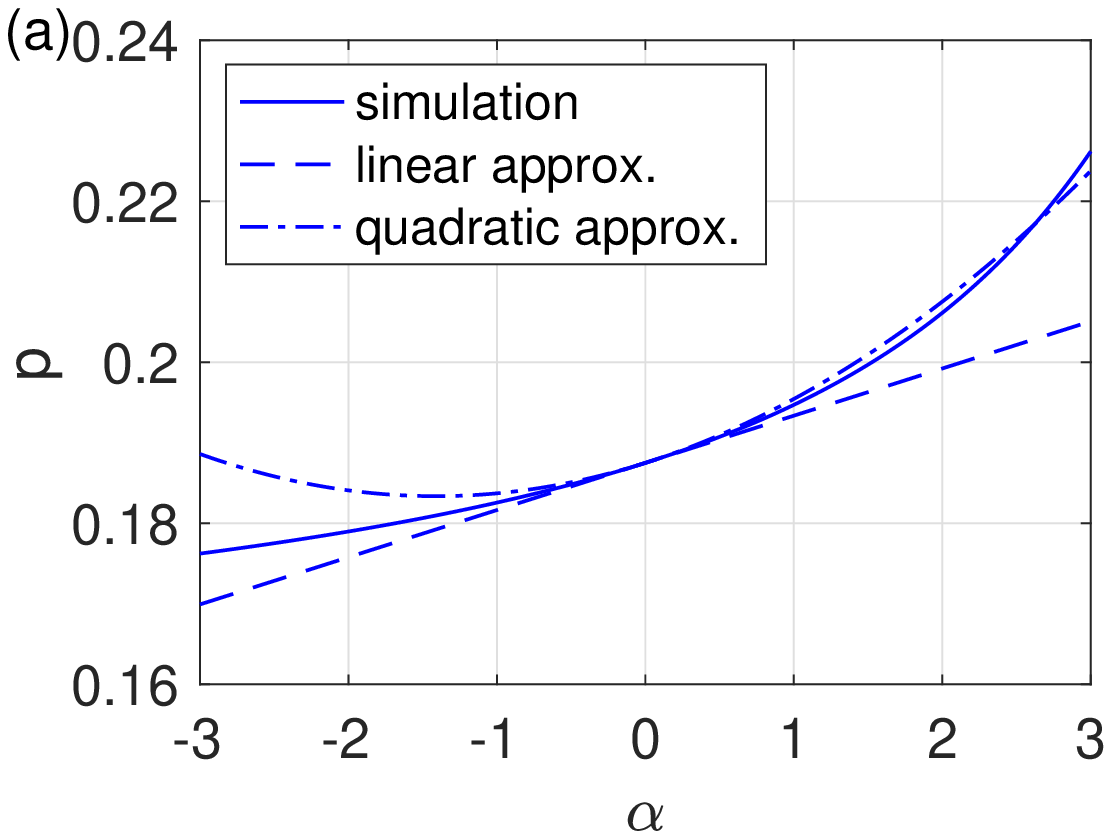}
\includegraphics[width=0.49\textwidth]{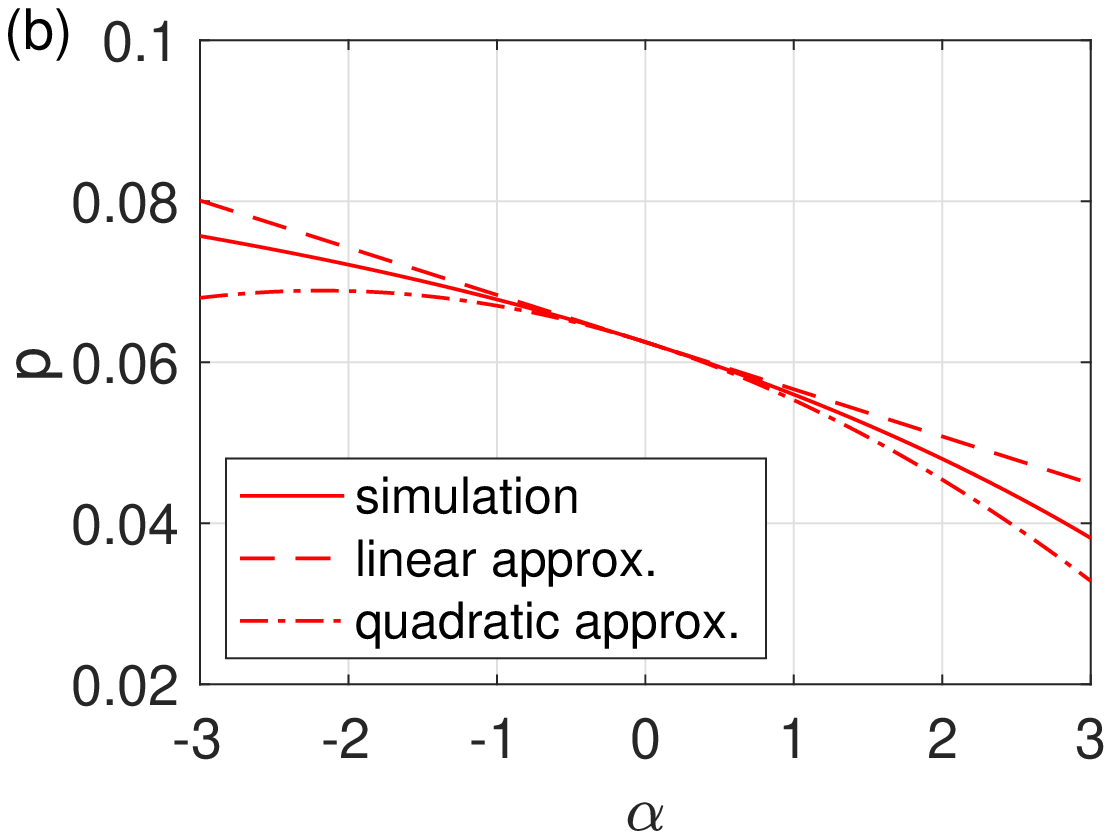}
\caption{{\bf Approximation of the stationary state fixed point solution in the weakly nonlinear regime}. For the nodes colored (a) blue and (b) red in the network illustrated in Figure~\ref{fig:01}, the stationary state fixed point solution obtained by simulations (solid curves) compared to the linear (dashed curves) and quadratic (dot-dashed curves) approximations obtained from Eq.~(\ref{eq:04:01}).}\label{fig:03}
\end{figure}

Next we compare the analytical approximations obtained for the stationary state fixed point in the weakly nonlinear regime to the true solution. We consider the primitive network illustrated in Figure~\ref{fig:01}, plotting in Figure~\ref{fig:03} the true solutions obtained from simulations using solid curves and comparing these to the linear and quadratic approximations, plotted with dashed and dot-dashed curves, respectively. Results are shown for the high- and low-degree nodes colored blue and red in the original network illustrations, which are plotted in panels (a) and (b), respectively. Results are plotted for $\alpha$ between $-3$ and $3$ and we have zoomed-in to each curve to see the finer features. In particular, we note that for $\alpha$ roughly in the range $(-1,1)$ both the linear and quadratic approximations capture the dynamics, and as $|\alpha|$ increases beyond this range the approximations begin to break down, as expected.

\section{Instability and Period-Doubling for Negative Bias}\label{sec:05}

We now turn our attention to the dynamics in the strongly nonlinear regime where more exotic nonlinear effects occur. We begin with sufficiently negative values of the bias parameter $\alpha$, where we see from Figure~\ref{fig:01} that the stationary state fixed point gives way to a periodic orbit. In particular, as $\alpha$ is decreased, the behavior of each $p_i$ in Figure~\ref{fig:01} bears the signature of a period-doubling bifurcation, which we will now characterize, at some critical bias value $\alpha_c$. To do so, we require the continuation of the stationary state fixed point beyond, where presumably the solution still exists but is unstable. We then seek to confirm that the solution in fact loses stability at $\alpha_c$ using a linear stability analysis.

To begin with this process, we require the Jacobian matrix $DF(\bm{p})$ of the iteration function $F(\bm{p})=\Pi(\bm{p})\bm{p}$, whose entries are given by
\begin{align}
DF_{ij}(\bm{p}) = \left\{\begin{array}{rl}\sum_{j=1}^N\frac{a_{ij}\alpha e^{\alpha p_i}\left[\left(\sum_{l=1}^Na_{lj}e^{\alpha p_l}\right)-e^{\alpha p_i}\right]}{\left(\sum_{l=1}^Na_{lj}e^{\alpha p_l}\right)^2}p_j&\text{if }i = j,\\
\frac{a_{ij}e^{\alpha p_i}}{\sum_{l=1}^Na_{lj}e^{\alpha p_l}} - \alpha\sum_{k=1}^N\frac{a_{ik}a_{kj}p_ke^{\alpha p_i}e^{\alpha p_j}}{\left(\sum_{l=1}^Na_{lk}e^{\alpha p_l}\right)^2}&\text{if }i\ne j.
\end{array}\right.
\end{align}
We note that $DF(\bm{p})$ is denser than $A$ or $\Pi(\bm{p})$, i.e., has more non-zero entries, due to the effect that an infinitesimal perturbation to one entry $i$ of $\bm{p}$ may have on another entry $j$ of $\bm{p}$ where $i$ and $j$ are not necessarily network neighbors. We then obtain the continuation of the stationary state fixed point solution $\bm{p}^*$ using Newton's method~\cite{Atkinson2008} to find the roots of the function $G(\bm{p})=F(\bm{p})-\bm{p}=\Pi(\bm{p})\bm{p}-\bm{p}$. Specifically, we begin at $\alpha=0$, obtain the solution $\bm{p}^*$ using Newton's method, then decrease $\alpha$ by a small amount and repeat. In panel (a) of Figure~\ref{fig:04} we present our results, using the network illustrated in Figure~\ref{fig:01} and plotting the solutions obtained from simulations using solid curves and the stationary-state fixed point solution obtained by Newton's method using dashed curves. Again, we highlight the simulation results for the high- and low-degree nodes colored blue and red in blue and red curves. We note that the results obtained by Newton's method perfectly match our observations from simulations in the weakly nonlinear regime, bisecting the period-doubling as $\alpha$ is decreased beyond the critical value $\alpha_c$, which we find below. 

\begin{figure}[htbp]
\centering
\includegraphics[width=0.95\textwidth]{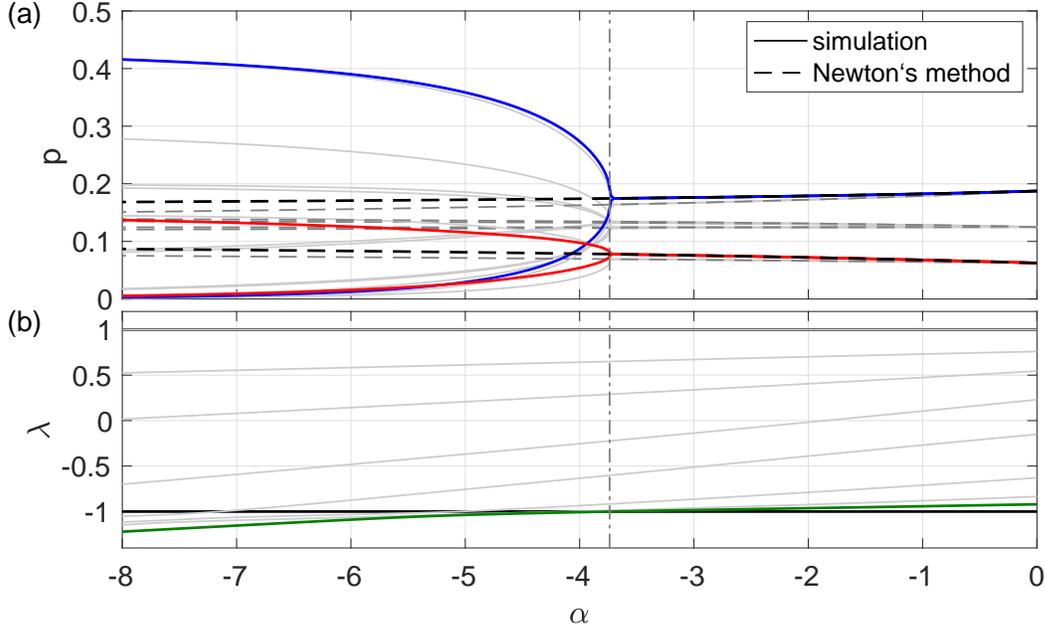}
\caption{{\bf Period-doubling in the nonlinear random walk}. Panel (a): For the network illustrated in Figure~\ref{fig:01}, the stationary state dynamics obtained from simulations (solid curves) and the stationary state fixed point obtained by Newton's methods (dashed curves). The probabilities for the nodes colored blue and red are plotted in blue and red. Panel (b): the eigenvalues of the Jacobian $DF(\bm{p}^*)$, where $\bm{p}^*$ is the stationary state fixed point obtained using Newton's method. The most negative eigenvalue is plotted in green. In both panels (a) and (b) the critical value $\alpha_c$ is denoted with a dot-dashed vertical line.}\label{fig:04}
\end{figure}

Next we seek to evaluate the stability of the stationary state fixed point solution $\bm{p}^*$ found using Newton's method. To do so, we inspect the eigenvalues of the Jacobian evaluated at $\bm{p}^*$, i.e., $DF(\bm{p}^*)$, denoting the $i^{\text{th}}$ eigenvalue of $DF(\bm{p}^*)$ as $\alpha$ is varied by $\lambda_i(\alpha)$. It follows that the stationary state fixed point $\bm{p}^*$ is stable if $|\lambda_i(\alpha)|<1$ for all $i=1,\dots,N$. For the nonlinear random walk, however, there is precisely one eigenvalue that is exactly one for all $\alpha$. This follows from the center manifold corresponding to scaling solutions: under the dynamics of Eq.~(\ref{eq:02:01}) and Eq.~(\ref{eq:02:04}) if $\bm{p}(t)$ is a non-zero solution, then so is $c\bm{p}(t)$ for any $c\in\mathbb{R}$. Because we consider the subspace $\Omega$ of probability vector in $\mathbb{R}^{N}$ we ignore this manifold since such scalings would yield solutions that are no longer in $\Omega$. 

Turning back to our example, we calculate the eigenvalues $\lambda_i(\alpha)$ of $DF(\bm{p}^*)$ as $\alpha$ is varied, where $\bm{p}^*$ is the stationary state fixed point obtained by Newton's method, and plot the results in panel (b) of Figure~\ref{fig:04}. First, we see that there is in fact one eigenvalue $\lambda(\alpha)=1$, while the other seven satisfy $\lambda_i(\alpha)<1$. (For this particular case we find that all eigenvalues are real, but in general eigenvalues may be complex since $DF(\bm{p}^*)$ is typically not symmetric.) Next, for $\alpha$ in the weakly nonlinear regime we find that all eigenvalues satisfy $\lambda(\alpha)_i>-1$, signifying that the stationary state fixed point is stable. However, when $\alpha$ is decreased to be sufficiently negative the most negative eigenvalue (plotted in green) crosses $-1$ so that one or more eigenvalues satisfy $\lambda_i(\alpha)<-1$, signifying that the stationary state fixed point loses stability. Moreover, the crossing of a real eigenvalue through $-1$ signifies that this is in fact a period-doubling bifurcation. Using these numerical results we then find $\alpha_c\approx-3.7375$ by evaluating at what $\alpha$ the most negative eigenvalue crosses $-1$. In both panels (a) and (b) we denote $\alpha_c$ with a dot-dashed vertical line, which matches perfectly with the onset of periodic orbits in our numerical simulations. Before proceeding to the next section, we note that while the specifics of the dynamics observed tend to depend on the underlying network structure, qualitatively similar transitions and period-doubling bifurcations can be found in the other networks we have considered (not shown).

\section{Multistability for Positive Bias}\label{sec:06}

Finally we turn our attention to the dynamics in the strongly nonlinear regime, but for positive values of the bias parameter. While in this region of parameter space we do not find any evidence of period-doubling or loss of stability of the fixed point solution, we do observe another phenomenon indicative of nonlinear dynamics: multistability. As seen in Figure~\ref{fig:01} a stable stationary state fixed point branch can be obtained simply by slowly increasing $\alpha$ as the dynamics are stepped forward. For the examples we have considered this remains true for values of $\alpha$ far larger than those shown in Figure~\ref{fig:01}. However, we also find that, for sufficiently large $\alpha$, other stable stationary state fixed point solutions can be found by varying the initial conditions for the nonlinear random walk.

\begin{figure}[htbp]
\centering
\includegraphics[width=0.94\textwidth]{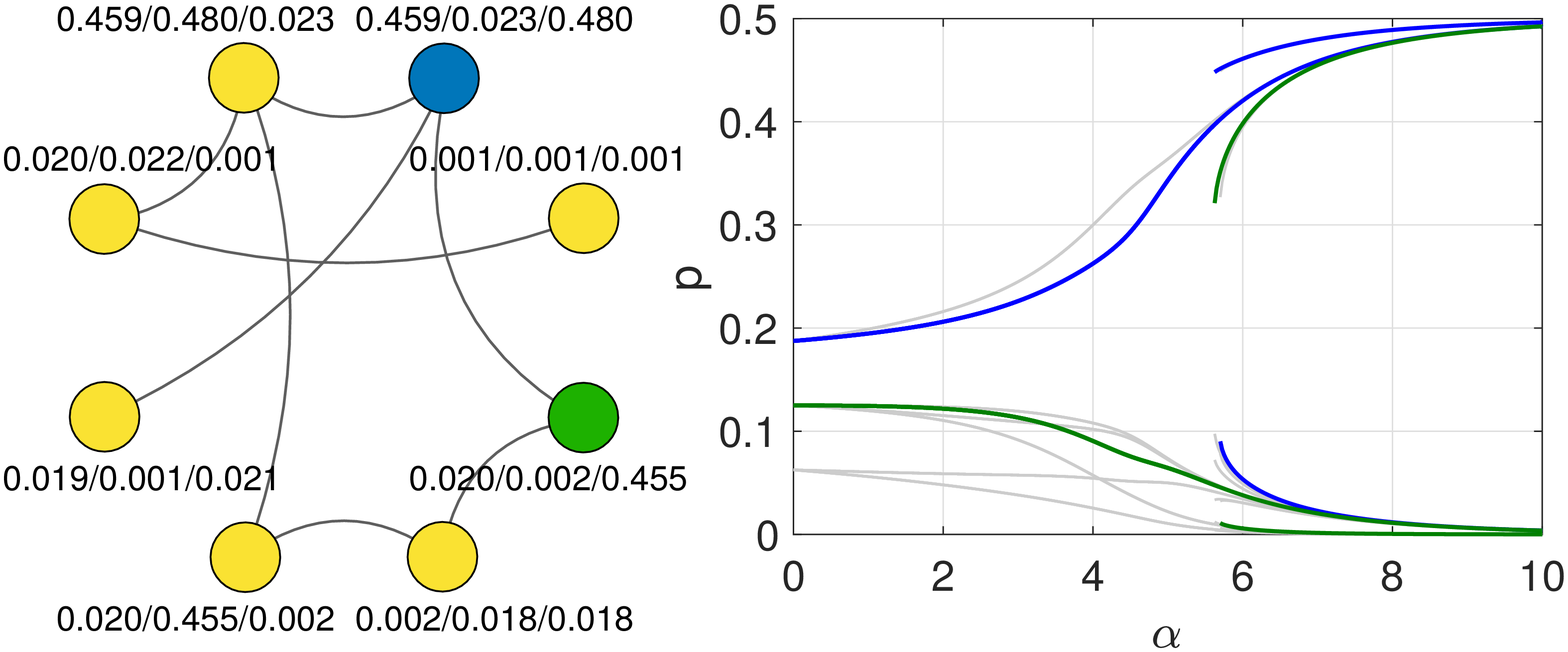}
\caption{{\bf Multistability in the nonlinear random walk}. Left: The network illustrated in Figure~\ref{fig:01} labeled with three stable stationary state fixed points of the nonlinear random walk for bias parameter $\alpha=7$. Right: the probabilities $p_i$ of the three stable stationary state fixed points as a function of the bias parameter $\alpha$. The probabilities of the nodes colored blue and green are plotted in blue and green, with other probabilities plotted in gray.}\label{fig:05}
\end{figure}

We demonstrate the emergence of multistability in the nonlinear random walk using the same network as that illustrated in Figure~\ref{fig:01}. In Figure~\ref{fig:05} we illustrate this network again, this time using different nodes in the network to highlight the features of the dynamics. We introduce three different initial conditions corresponding to starting all the random walkers at the same node, i.e., setting all $p_i(0)=0$ except for one node with $p_i(0)=1$. Specifically, we consider starting all the random walkers at the top-right node, the next node clockwise to this node, and finally the node at the bottom left. For each of the initial conditions we then iterate the dynamics starting at $\alpha=10$ and decrease $\alpha$ in small increments after reaching steady-state at each value considered. We then plot the result of the three different simulations, highlighting the nodes colored blue and green with blue and green curves. Other probabilities are plotted in gray. In particular, the dynamics for each of the three different initial conditions converge to different stable stationary state fixed points. To better match the dynamics to the network structure, we then label each node with the three different stationary state fixed point values taken for $\alpha=7$. 

The dynamics observed in this example contain several interesting nonlinear features. First, we only observe multistability for sufficiently large $\alpha$ values. In fact, as $\alpha$ is decreased from $10$, the three different branches eventually collapse on top of one another near $\alpha\approx 5.7$. Moreover, this collapse occurs discontinuously: while one of the branches is precisely the branch obtained by continually increasing $\alpha$ (and corresponds to that seen in Figure~\ref{fig:01}), the two others snap back to this continuous branch in what appears to be a discontinuous jump. Second, the values of the specific nodes in the network vary between the three stationary state fixed points found. For instance, along the continuous branch the two nodes with largest degree ($k=3$) have the two largest probabilities, but this is not true in the two other solutions. Instead, each of the high degree nodes take on a relatively small probability in steady-state for one of the two cases while the larger probability is transferred to another node of smaller degree. Similar to the dynamics in the strongly nonlinear regime with negative bias, while the specifics of the dynamics vary greatly depending on the underlying network structure, we find that in other network structure (not shown) qualitatively similar properties persist, including multistability only for sufficiently large $\alpha$, discontinuous branches, and interesting patterns of steady-state probabilities. 

To explore more deeply the properties of the multistability described above we employ some broader numerical searches of the state-space for the nonlinear random walk and characterize basin stability properties of each fixed point found. For a general dynamical system with an attractor denoted $a$ the basin stability of $a$ describes the effective size (via a measure) of the basin of attraction for $a$, denoted $\mathcal{B}(a)$~\cite{Schultz2017NJP}. Thus, the basin stability of $a$, denoted $S_{\mathcal{B}(a)}$ can be written
\begin{align}
S_{\mathcal{B}(a)}=\mu(\mathcal{B}(a)),\label{eq:06:01}
\end{align}
where $\mu$ is an appropriately defined measure over an appropriately chosen domain that contains the collection of basins of attraction of all attractors. Often, $\mu$ is taken to to be proportional to a volume so that $\mu(\mathcal{B}(a))$ is the fraction of the state space corresponding to $\mathcal{B}(a)$. For the specific case of nonlinear random walks studied here, we take $\mu$ to be the normalized volume of $\Omega$ defined in Eq.~(\ref{eq:02:02}) (note that $\Omega$ is a bounded, compact $N-1$-dimensional subset of $\mathbb{R}^N$) so that $S_{\mathcal{B}(a)}=\mu(\mathcal{B}(a))$ gives the fraction of $\Omega$ occupied by $B(a)$. Moreover, for non-negative $\alpha$ we only find attractors that are fixed points, i.e., $a=\bm{p}^*$.

\begin{figure}[htbp]
\centering
\includegraphics[width=0.549\textwidth]{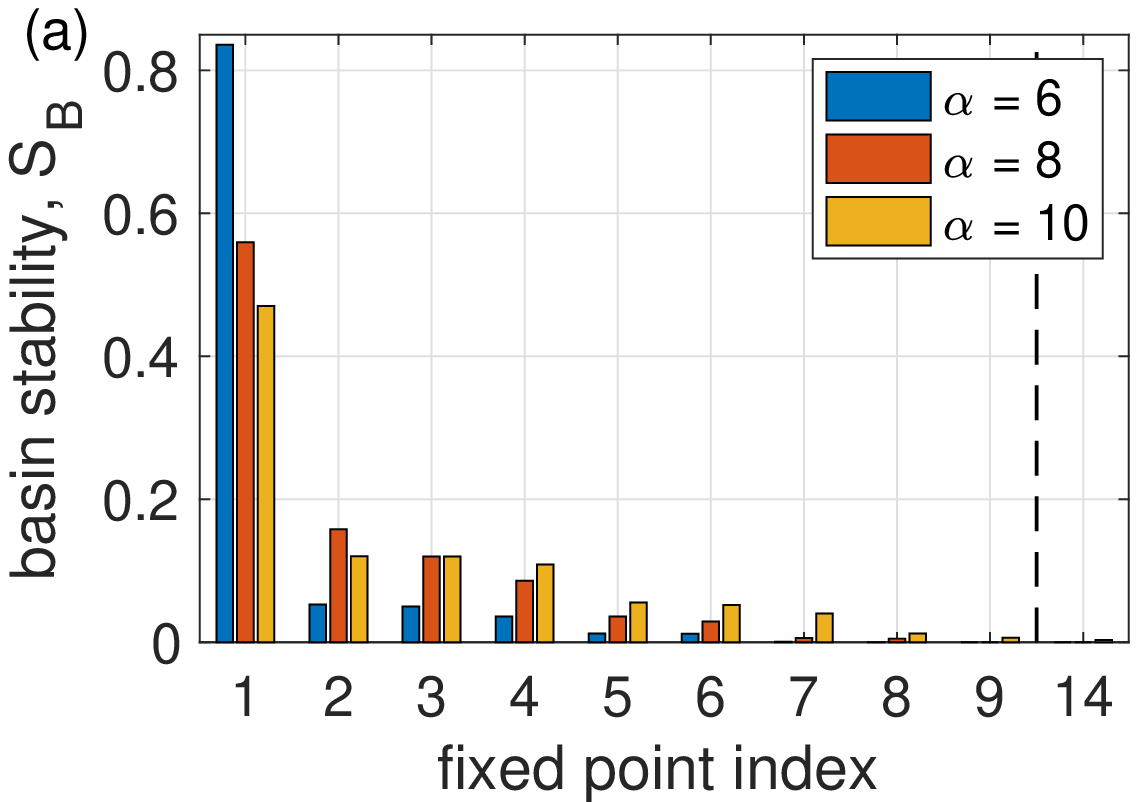}
\includegraphics[width=0.444\textwidth]{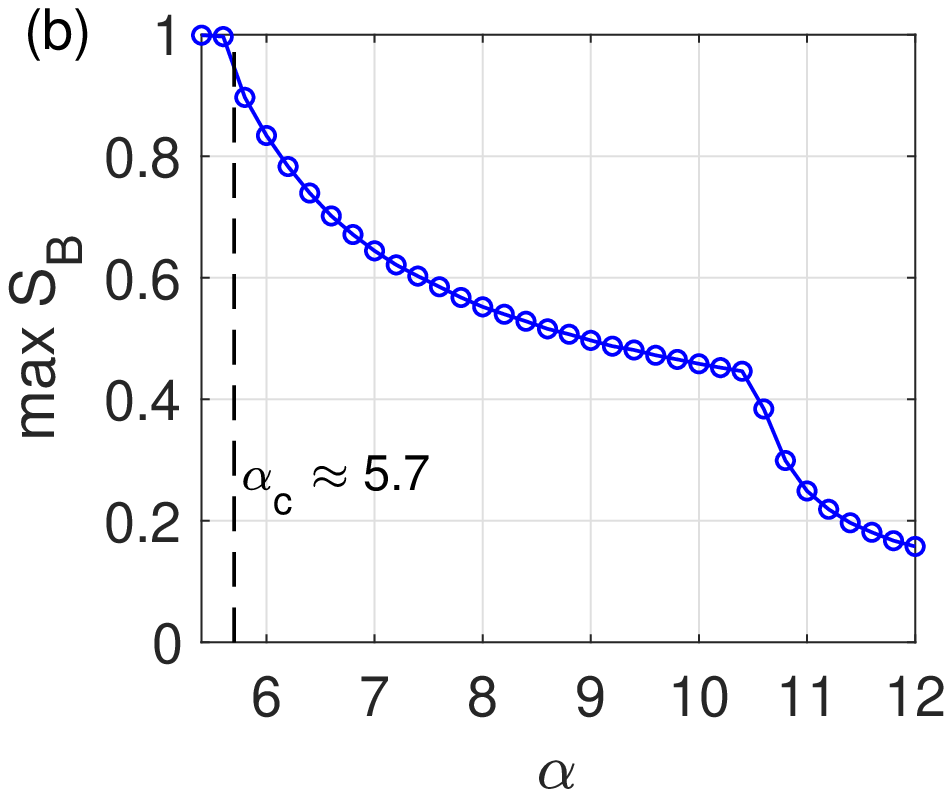}
\caption{{\bf Basin stability in the multistable regime}. (a) The basin stability $S_{\mathcal{B}}$ for the stable fixed points identified in the nonlinear random walk with $\alpha=6$ (blue), $8$ (red), and $10$ (yellow) on the $N=8$ node network illustrated in Figure~\ref{fig:01}. Fixed points were identified by iterating $10^4$ uniformly-distributed random initial conditions in $\Omega$. Note: $9$ fixed points were identified for $\alpha=6$ and $8$ and $14$ were identified for $\alpha=14$. (b) The largest basin stability $S_{\mathcal{B}}$ as a function of the bias parameter $\alpha$.}\label{fig:06}
\end{figure}

We then proceed numerically, identifying fixed points using a large set of randomly chosen initial conditions (drawn uniformly from $\Omega$) and iterating to steady-state. For each distinct fixed point $\bm{p}^*$ found we then approximate its basin stability $S_{\mathcal{B}(\bm{p}^*)}$ as the fraction of all initial conditions that ended up at $\bm{p}^*$. In Figure~\ref{fig:06} (a) we present our results obtained for the network first presented in Figure~\ref{fig:01} using $10^4$ different randomly chosen initial conditions each for $\alpha=6$, $8$, and $10$, plotted in blue, orange, and yellow, respectively. For each value of $\alpha$ we have indexed the steady-state fixed points in order of decreasing basin stability. Starting with $\alpha=6$, which occurs relatively close to the point where we first observe multistability, there is one fixed point whose basin of attraction takes up the vast majority (over $4/5$) of $\Omega$. We also find an additional $7$ stable fixed points, each of which has a basin of attraction that occupies a very small fraction of $\Omega$. When $\alpha$ is increased, we find that the largest basin of attraction shrinks while the others increase in size. This is particularly true for $\alpha=10$, where we also find a large number of fixed points, $14$ in total, compared to $8$ for both $\alpha=6$ and $8$. While it is not guaranteed that this numerical method finds each stable fixed point, these results suggest that for more extreme values of $\alpha$ the relative sizes of the basins of attraction become less heterogeneous and it is possible that more stable fixed point may exist.

Lastly, to complement the investigation into the basins of attraction for the full set of stable fixed points of the nonlinear random walk, we study in more detail the largest basin of attraction in the system. In Figure~\ref{fig:06} (b) plot the basin stability $S_{\mathcal{B}}$ for the fixed point with the largest basin of attraction as $\alpha$ is varied. For $\alpha$ less than the critical value $\alpha_c\approx 5.7$ given above, we find a single fixed point whose basin of attraction contains all $10^4$ randomly chosen initial conditions used., yielding $\text{max}~S_{\mathcal{B}}=1$. However, beyond this critical value, $\text{max}~S_{\mathcal{B}}$ appears to decay smoothly. Interestingly, near $\alpha\approx 10.5$ we observe a second significant dip in $\text{max}~S_{\mathcal{B}}$ which appears to occur due to a sharp increase in the number of stable fixed point of the system, effectively detracting from the size of the largest basin of attraction. We emphasize, however, that the results presented here represent numerical estimates of basin stability and correspond specifically to the network illustrated in Figure~\ref{fig:01}. Therefore, the general properties of multistability that emerges in nonlinear random walks requires deeper investigation in future work.

\section{Discussion}\label{sec:07}

In this paper we have studied the dynamics that emerge in nonlinear random walks on complex networks. The nonlinear random walk represents a generalization to the typical linear random walk, where the transition probabilities dictating the evolution of the system remain constant. Instead, in the case of the nonlinear random walk the transition probabilities to depend on the current state of the system, so the nonlinear random walk is an instance of a specific nonlinear Markov chain. The complexity in the dynamics of the nonlinear random walk thus results in a potentially better model for describing more complicated transport and other dynamics where resources or other quantities tend to be preferentially routed towards nodes depending on their current state. Specifically, the dynamics depend on a biasing parameter that allows for random walkers to be routed more towards nodes that currently hold a larger or smaller fraction of the random walkers in the system. This biasing parameter comes into the transition probability through an exponential function, where, if nodes $i$ and $j$ are network neighbors, then the probability of moving to node $i$ given that one is at node $j$ is proportional to the exponential of the fraction of random walkers at node $i$ times the bias parameter, denoted by $\alpha$. One motivating example is the transport of resources through a network, where institutions may prefer to invest more heavily in other poorer or wealthier institutions, depending on their humanitarian or capitalistic values, however we believe that the simplicity of the model lends itself towards other examples, perhaps with some additional properties built in depending on the specific application.

Our initial analysis of the nonlinear random walk has roughly been split into two portions: the weakly nonlinear regime and the strongly nonlinear regime. The weakly nonlinear regime consists of values of the bias parameter $\alpha$ sufficiently close to zero, where the choice $\alpha=0$ recovers the standard linear random walk on the network. First we proved for networks with primitive adjacency matrices that for sufficiently small values of the biasing parameter the dynamics converge to a unique stationary state fixed point analogous to that of the linear random walk. We then presented a perturbative analysis of the stationary state fixed point in this regime, finding an asymptotic approximation of the solution. In the strongly nonlinear regime we found different dynamics emerging for positive and negative biasing. When the bias parameter is sufficiently negative we characterize a period-doubling bifurcation where the stationary state fixed point loses stability and gives rise to a periodic orbit. In the framework of the transport of resources, this suggests that routing resources too strongly towards poorer institutions yields an instability causing the wealth of all institutions to fluctuate indefinitely. When the bias parameter is sufficiently positive we discovered a multistability in the dynamics where several stable stationary state fixed points emerge. Using the concept of basin stability, we then showed that there is one fixed point whose basin of attraction is much larger than the others', but as biasing is made more extreme this dominant basin of attraction shrinks, giving way to other basins of attraction. In the framework of the transport of resources, this multistability suggests that routing resources too strongly towards wealthier institutions results in a red herring effect where wealthy institutions stay wealthy because they already are wealthy. Overall, the nonlinear random walk displays a wide array of dynamical phenomena ranging from the relaxation-type dynamics in the weakly nonlinear regime that are analogous to the linear system to more exotic behaviors in the strongly nonlinear regime that are signatures of the system's nonlinearity.

As noted previously, nonlinear Markov chains have been investigated in a number of contexts, but to date little work has examined from a nonlinear dynamics perspective the behavior that arises in nonlinear random walks. We have used a handful of simple fixed networks to highlight the dynamics that emerge, so it remains to be seen how the dynamics depend on structural properties of the network such as size, overall connectivity, heterogeneity, directedness, etc. Additionally, we examined a family of nonlinear random walks that use an exponential function to define transition probabilities, but in principle a number of other functions may be used. As different network properties and transition rules are explored, further questions that may be of interest include: How large is the weakly nonlinear regime as a function of the network structure? How does the critical bias parameter corresponding to period-doubling behave for different networks? Is there always a period-doubling bifurcation, or are there networks for which no such period-doubling occurs? Can higher-order periodic orbits be found? How does the network structure affect the multistability for positive bias? What features of the network give rise to more or fewer stationary state fixed points that are all simultaneously stable? Does multistability occur in all networks? These questions and more yield a rich topic for dynamical systems and complex networks researchers to study in years to come.

\appendix

\section{Contraction Mappings and the Contraction Mapping Theorem}\label{app:A}

\begin{definition}[Contraction Mapping~\cite{Hunter2001}]\label{def:Contraction}
Let $X$ be a metric space equipped with a metric $d$. A mapping $T:X\to X$ is a contraction mapping if there exists some constant $0\le c<1$ such that 
\begin{align}
d(T(x),T(y))\le c d(x,y)\label{eq:A:01}
\end{align}
for all $x,y\in X$.
\end{definition}

\begin{theorem}[Contraction Mapping Theorem~\cite{Hunter2001}]\label{Theorem:CMT}
Let $T:X\to X$ be a contraction mapping on a complete metric space $X$. Then there exists a unique fixed point $x^*\in X$. Moreover, the iteration $x_{n+1}=Tx_n$ converges to $x^*$, that is, $\lim_{n\to\infty}x_n=\lim_{n\to\infty}T^nx_0=x^*$ for any initial $x_0\in X$.
\end{theorem}

\bibliographystyle{siam}
\bibliography{references}

\begin{thebibliography}{10}

\bibitem{Atkinson2008}
{\sc K.~E. Atkinson}, {\em An Introduction to Numerical Analysis}, John Wiley
  \& Sons, 2008.

\bibitem{Brin1998}
{\sc S.~Brin and L.~Page}, {\em The anatomy of a large-scale hypertextual web
  search engine}, Computer Networks and ISDN Systems, 30 (1998), pp.~107 --
  117.
\newblock Proceedings of the Seventh International World Wide Web Conference.

\bibitem{Butkovsky2014SIAM}
{\sc O.~Butkovsky}, {\em On ergodic properties of nonlinear markov chains and
  stochastic mckean--vlasov equations}, Theory of Probability \& Its
  Applications, 58 (2014), pp.~661--674.

\bibitem{Durrett2016}
{\sc R.~Durrett and R.~Durrett}, {\em Essentials of stochastic processes},
  Springer, 2016.

\bibitem{Frank2008JPA}
{\sc T.~Frank}, {\em Markov chains of nonlinear markov processes and an
  application to a winner-takes-all model for social conformity}, Journal of
  Physics A: Mathematical and Theoretical, 41 (2008), p.~282001.

\bibitem{Frank2008PLA}
\leavevmode\vrule height 2pt depth -1.6pt width 23pt, {\em Nonlinear markov
  processes: deterministic case}, Physics Letters A, 372 (2008),
  pp.~6235--6239.

\bibitem{Frank2009EPJB}
\leavevmode\vrule height 2pt depth -1.6pt width 23pt, {\em Deterministic and
  stochastic components of nonlinear markov models with an application to
  decision making during the bailout votes 2008 (usa)}, The European Physical
  Journal B, 70 (2009), pp.~249--255.

\bibitem{Frank2011BJP}
\leavevmode\vrule height 2pt depth -1.6pt width 23pt, {\em Stochastic processes
  and mean field systems defined by nonlinear markov chains: An illustration
  for a model of evolutionary population dynamics}, Brazilian Journal of
  Physics, 41 (2011), p.~129.

\bibitem{Frank2013ISRN}
\leavevmode\vrule height 2pt depth -1.6pt width 23pt, {\em Strongly nonlinear
  stochastic processes in physics and the life sciences}, ISRN Mathematical
  Physics, 2013 (2013), p.~149169.

\bibitem{Gleich2015SIAMRev}
{\sc D.~F. Gleich}, {\em Pagerank beyond the web}, SIAM Review, 57 (2015),
  pp.~321--363.

\bibitem{Gomez2008PRE}
{\sc J.~G\'omez-Garde\~nes and V.~Latora}, {\em Entropy rate of diffusion
  processes on complex networks}, Phys. Rev. E, 78 (2008), p.~065102.

\bibitem{Gorenflo2002CP}
{\sc R.~Gorenflo, F.~Mainardi, D.~Moretti, G.~Pagnini, and P.~Paradisi}, {\em
  Discrete random walk models for space--time fractional diffusion}, Chemical
  physics, 284 (2002), pp.~521--541.

\bibitem{Hunter2001}
{\sc J.~K. Hunter and B.~Nachtergaele}, {\em Applied analysis}, World
  Scientific Publishing Company, 2001.

\bibitem{Kolokoltsov2010}
{\sc V.~N. Kolokoltsov}, {\em Nonlinear Markov processes and kinetic
  equations}, vol.~182, Cambridge University Press, 2010.

\bibitem{Kolokoltsov2012IJSP}
\leavevmode\vrule height 2pt depth -1.6pt width 23pt, {\em Nonlinear markov
  games on a finite state space (mean-field and binary interactions)},
  International Journal of Statistics and Probability, 1 (2012), p.~77.

\bibitem{MacCluer2000SIAMRev}
{\sc C.~R. MacCluer}, {\em The many proofs and applications of perron's
  theorem}, SIAM Review, 42 (2000), pp.~487--498.

\bibitem{Meyer2000}
{\sc C.~D. Meyer}, {\em Matrix analysis and applied linear algebra}, Siam,
  2000.

\bibitem{Nicosia2017PRL}
{\sc V.~Nicosia, P.~S. Skardal, A.~Arenas, and V.~Latora}, {\em Collective
  phenomena emerging from the interactions between dynamical processes in
  multiplex networks}, Phys. Rev. Lett., 118 (2017), p.~138302.

\bibitem{Noh2004PRL}
{\sc J.~D. Noh and H.~Rieger}, {\em Random walks on complex networks}, Phys.
  Rev. Lett., 92 (2004), p.~118701.

\bibitem{Page1999}
{\sc L.~Page, S.~Brin, R.~Motwani, and T.~Winograd}, {\em The pagerank citation
  ranking: Bringing order to the web.}, Technical Report 1999-66, Stanford
  InfoLab, 1999.

\bibitem{Rosvall2008PNAS}
{\sc M.~Rosvall and C.~T. Bergstrom}, {\em Maps of random walks on complex
  networks reveal community structure}, Proceedings of the National Academy of
  Sciences, 105 (2008), pp.~1118--1123.

\bibitem{Saburov2016NLA}
{\sc M.~Saburov}, {\em Ergodicity of nonlinear markov operators on the finite
  dimensional space}, Nonlinear Analysis: Theory, Methods \& Applications, 143
  (2016), pp.~105--119.

\bibitem{Schultz2017NJP}
{\sc P.~Schultz, P.~J. Menck, J.~Heitzig, and J.~Kurths}, {\em Potentials and
  limits to basin stability estimation}, New Journal of Physics, 19 (2017),
  p.~023005.

\bibitem{Sinatra2011PRE}
{\sc R.~Sinatra, J.~G\'omez-Garde\~nes, R.~Lambiotte, V.~Nicosia, and
  V.~Latora}, {\em Maximal-entropy random walks in complex networks with
  limited information}, Phys. Rev. E, 83 (2011), p.~030103.

\end{thebibliography}

\end{document}